\documentclass[american,aps,prl,reprint]{revtex4-2}
\usepackage[T1]{fontenc}
\usepackage[latin9]{inputenc}
\usepackage{color}
\usepackage{babel}
\usepackage{amsmath}
\usepackage{amsthm}
\usepackage{amssymb}
\usepackage[unicode=true,pdfusetitle,
 bookmarks=true,bookmarksnumbered=false,bookmarksopen=false,
 breaklinks=false,pdfborder={0 0 1},backref=false,colorlinks=true]
 {hyperref}
\hypersetup{
 pdfborderstyle=,allcolors=magenta}

\makeatletter
\theoremstyle{plain}
\newtheorem{thm}{\protect\theoremname}
\newtheorem{lem}[thm]{\protect\lemmaname}

\usepackage{color}
\usepackage{babel}
\usepackage{physics}

\usepackage{times}
\usepackage{txfonts}
\usepackage{braket}
\usepackage{colortbl}
\definecolor{myurlcolor}{rgb}{0,0,0.7}

\makeatother

\providecommand{\theoremname}{Theorem}
\providecommand{\lemmaname}{Lemma}

\begin{document}
\title{Is there a finite complete set of monotones in any quantum resource
theory?}

\author{Chandan Datta}
\email{dattachandan10@gmail.com}
\author{Ray Ganardi}
\author{Tulja Varun Kondra}
\author{Alexander Streltsov}
\email{a.streltsov@cent.uw.edu.pl}

\affiliation{Centre for Quantum Optical Technologies, Centre of New Technologies,
University of Warsaw, Banacha 2c, 02-097 Warsaw, Poland}

\begin{abstract}
  Entanglement quantification aims to assess the value of quantum states for quantum information processing tasks.
  A closely related problem is state convertibility, asking whether two remote parties can convert a shared quantum state into another one without exchanging quantum particles.
  Here, we explore this connection for quantum entanglement and for general quantum resource theories.
  For any quantum resource theory which contains resource-free pure states, we show that there does not exist a finite set of resource monotones which completely determines all state transformations.
  We discuss how these limitations can be surpassed, if discontinuous or infinite sets of monotones are considered, or by using quantum catalysis.
  We also discuss the structure of theories which are described by a single resource monotone and show equivalence with totally ordered resource theories.
  These are theories where a free transformation exists for any pair of quantum states.
  We show that totally ordered theories allow for free transformations between all pure states.
  For single-qubit systems, we provide a full characterization of state transformations for any totally ordered resource theory.
\end{abstract}

\maketitle

Entangled quantum systems can exhibit features which seem to contradict
our intuition, based on our ``classical'' perception of nature~\cite{HorodeckiRevModPhys.81.865}.
Even Einstein was puzzled by some of the consequences of entanglement,
concluding that quantum theory cannot be complete~\cite{EinsteinPhysRev.47.777}. Today,
entangled quantum systems are actively explored as an important ingredient
of the emerging quantum technologies~\citep{HorodeckiRevModPhys.81.865}.
This includes applications such as quantum key distribution \cite{Ekert_PhysRevLett.67.661},
where entangled systems are used to establish a provably secure key
for communication between distant parties. Another groundbreaking
application of entanglement is quantum teleportation \cite{Bennett_PhysRevLett.70.1895}, allowing
to send the state of a quantum system to a remote party by using shared
entanglement and classical communication.

The development of a resource theory of entanglement~\citep{HorodeckiRevModPhys.81.865}
made it possible to study the role of entanglement for technology
in a systematic way. This theory introduced the distant lab paradigm,
with two remote parties (Alice and Bob) being equipped with local
quantum laboratories, and connected via a classical communication
channel \cite{Bennett_purification_PhysRevLett.76.722,BennettPhysRevA.54.3824,Bennett_asymptotic_PhysRevA.53.2046}. It has been noticed that entanglement between
Alice and Bob cannot be created in this setting. Thus, entangled states
become a valuable resource, allowing the remote parties to perform
tasks which are not possible without it.

In recent years, it became clear that not all quantum technological
tasks are based on entanglement, but can make use of other quantum
features, such as quantum coherence~\cite{BaumgratzPhysRevLett.113.140401,StreltsovRevModPhys.89.041003}, contextuality \cite{Kochen1967,Howard2014,Budroni_2021}, or imaginarity
\cite{Hickey_2018,WuPhysRevLett.126.090401,WuPhysRevA.103.032401,Renou_2021}. This has led to the development of general quantum resource
theories~\citep{ChitambarRevModPhys.91.025001}. In analogy to entanglement,
a quantum resource theory is based on the set of free states $\{\rho_{f}\}$
and free operations $\{\Lambda_{f}\}$. All states which are not free are called resource states. A free operation cannot create resource states from free states. The sets of free states and operations can be motivated
by physical constraints, as is done e.g. in the resource theory of
quantum thermodynamics~\citep{BrandaoPhysRevLett.111.250404,Goold_2016},
where the free state is the Gibbs state, and the free operations preserve
the total energy of the system and a heat bath~\citep{Janzing2000}.
Another motivation for a resource theory can arise from symmetries,
where the free states and operations are symmetric with respect to
some physical transformations. An example for such theory is the
resource theory of asymmetry~\citep{Gour_2008}. Also, the resource theory of coherence can be formulated in this framework, if the free states are diagonal in a reference basis, and the free operations are dephasing covariant~\cite{ChitambarPhysRevLett.117.030401,ChitambarPhysRevA.94.052336,ChitambarPhysRevA.95.019902,MarvianPhysRevA.94.052324,ChitambarPhysRevA.97.050301}. Similarly, the resource theory of imaginarity has free states which have only real elements in a reference basis, and the free operations are covariant with respect to transposition~\cite{Kondra2022}.

Two fundamental problems in any quantum resource theory are \emph{state
convertibility} and \emph{resource quantification}. The state convertibility
problem is asking whether for two quantum states there exists a free
operation converting one state into the other. The goal of resource
quantification is to quantify the amount of the resource in a quantum
state. In general, there is no unique quantifier which captures all
aspects of a resource theory, and a suitable quantifier depends on the
concrete problem under study.

There are some elementary properties which are common to all resource
quantifiers~\citep{ChitambarRevModPhys.91.025001}. Recalling that
resource states cannot be created from free states via free operations, it is intuitive
to assume that the degree of the resource in a quantum system cannot
increase under free operations, even if the initial state is not free.
Thus, every meaningful resource quantifier should not increase under free operations~\citep{BennettPhysRevA.54.3824,VedralPhysRevLett.78.2275,Vidal2000,ChitambarRevModPhys.91.025001}:
\begin{equation}
R(\Lambda_{f}[\rho])\leq R(\rho)\label{eq:Monotone-1}
\end{equation}
for any state $\rho$ and any free operation $\Lambda_{f}$. Quantifiers
having this property are also called \emph{resource monotones}.

Both problems mentioned above -- state convertibility and resource
quantification -- are in fact closely connected. A state $\rho$
can be converted into $\sigma$ via free operations if and only if
\begin{equation}
R(\rho)\geq R(\sigma)\label{eq:Monotone-2}
\end{equation}
holds true for all resource monotones~\citep{TakagiPhysRevX.9.031053}.
On the other hand,
the fact that Eq.~(\ref{eq:Monotone-2}) holds for some resource
monotone $R$ does not guarantee that the transformation $\rho\rightarrow\sigma$
is possible via free operations. There might however exist a \emph{complete
set of resource monotones} $\{R_{i}\}$ which completely characterizes
all state transformations, i.e., a transformation $\rho\rightarrow\sigma$ is possible if and only if $R_i(\rho) \geq R_i(\sigma)$ holds true for all $i$. The first such complete set of monotones
has been presented for bipartite pure states in entanglement theory~\citep{PhysRevLett.83.436,VidalPhysRevLett.83.1046},
and it was shown that there is no finite set of faithful and strongly monotonic entanglement monotones
which can capture transformations between all mixed states~\citep{GourPhysRevA.72.022323}.
Complete sets of monotones for concrete resource theories have been studied~\cite{PhysRevLett.122.140403,PhysRevX.8.021033,PhysRevA.95.062314,Gour2018,PhysRevLett.125.180505}, and constructions for general quantum resource theories have been presented in~\citep{TakagiPhysRevX.9.031053}. It is worth noting that quantum resource theories
that are completely governed by a majorization relation have a finite set of monotones~\citep{majorization_complete, majorization_complete_1}.

\medskip

\textbf{\emph{Finite sets of resource monotones cannot be complete.}} In this article
we show that a finite complete set of resource monotones does not
exist for a large class of quantum resource theories. Our results
make only minimal assumptions on the resource monotones: additionally
to Eq.~(\ref{eq:Monotone-1}) we require that the resource monotones
are \emph{continuous} (A resource monotone $R$ is continuous if for all states $\rho$ and $\varepsilon>0$, there exists a $\delta>0$ such that for all $\sigma$ that satisfies $||\rho-\sigma||_1<\delta$, we have $|R(\rho)-R(\sigma)|<\varepsilon$, where $||M||_{1}=\mathrm{Tr}\sqrt{M^{\dagger}M}$
is the trace norm) and \emph{faithful} (A resource monotone $R$ is faithful if $R(\rho)=0$ if and only if $\rho$ is a free state).
Continuity is a very natural assumption which is
fulfilled for most resource monotones studied in the literature --- it guarantees that the value of a resourceful state is robust to perturbations.
In fact, in many cases the monotones fulfill continuity in an even stronger
form, e.g. many entanglement monotones are asymptotically continuous~\citep{HorodeckiPhysRevLett.84.2014,Winter2016}.
Similarly, faithful monotones are often preferred since they detect some value in any non-free state.
We also use the standard assumptions that the set of free states is convex and compact, that the identity operation is free, and that any free state can be obtained from any state via free operations (this is fulfilled by resource theories that ``admit a tensor product structure''~\cite{ChitambarRevModPhys.91.025001}  since for any free state $\sigma$, the following measure-and-prepare channel is also free: $\Lambda(\rho)=\Tr(\rho)\sigma$).
The latter assumption implies that any resource monotone is minimal and constant on all free states -- without loss of generality we set it to zero.
We further say that a state $\rho$ can be converted into a state $\sigma$ via free operations if for any $\varepsilon > 0$ there is a free operation $\Lambda_f$ such that $||\Lambda_{f}(\rho)-\sigma||_{1}<\varepsilon$.
Clearly, the trivial resource theory where all states and all operations are free admits a complete set of continuous monotones.
Therefore we say that a resource theory is non-trivial if there exists a free state and a non-free state.
With these assumptions, we are now ready to prove
the first main result of this article.
\begin{thm}
\label{thm:Entanglement}For any non-trivial resource theory which contains free
pure states, there does not exist a finite complete set of continuous and faithful resource
monotones.
\end{thm}

\begin{proof}
By contradiction, let there be a complete finite set of continuous resource monotones $\{R_{i}\}$.
Let $\rho$ be a non-free state.
Since the set of free states is compact (and therefore closed), without loss of generality we can assume that $\rho$ is full rank
-- otherwise we take a mixture with the completely mixed state.
Since $R_i$ is faithful, we have $R_i \pqty{\rho} > 0$ for all $i$.
Moreover, we define the pure state
\begin{equation}
\ket{\psi_{\varepsilon}}=\sqrt{1-\varepsilon}\ket{\phi_{f}}+\sqrt{\varepsilon}\ket{\phi_{f}^{\perp}}
\end{equation}
with some free pure state $\ket{\phi_{f}}$ and $0<\varepsilon<1$.
Using again the fact that the set of free states is closed, the state $\ket{\psi_{\varepsilon}}$
can be chosen such that it is not free for all small $\varepsilon > 0$, i.e. we choose $\ket{\phi_f}$ to be on the boundary of the set of free states.
Since $R_i$ is continuous and $R_i (\psi_{\varepsilon = 0}) = 0$, we can choose $\varepsilon_i > 0$ such that $R_{i}(\rho)\geq R_{i}(\psi_{\varepsilon_i})$ for each $i$.
Take $\varepsilon = \min_i \varepsilon_i$, which must be strictly positive since there are a finite number of $R_i$.
Using again the continuity of $R_i$, we have $R_i \pqty{\rho} \geq R_i \pqty{\psi_\varepsilon}$ for all $i$.
If $\{R_{i}\}$ form a complete set of monotones, there must be a free operation converting $\rho$ into $\ket{\psi_{\varepsilon}}$.
Note that $\ket{\psi_{\varepsilon}}$ is a resource state and that $\rho$ is full rank.
It is however not possible to convert a full rank state into a pure resource state via free operations~\cite{FangPhysRevLett.125.060405,PhysRevA.101.062315}, see also Supplemental Material.
We thus arrive at a contradiction, and the proof is complete.
\end{proof}
The above theorem applies to the resource theory of entanglement,
both in bipartite and multipartite setting. Moreover, the resource
theories of coherence, asymmetry, and imaginarity also contain resource-free
pure states, which makes our theorem applicable also to these theories.
The theorem also applies to the resource theory of quantum thermodynamics
in the limit $T\rightarrow0$ if the ground state of the corresponding
Hamiltonian is not degenerate, since the Gibbs state is pure in this
case.

As a particular example, this means that no finite collection of continuous and faithful monotones can characterize the state transitions in PPT theory.
This is despite the fact that for any given two states $\rho, \sigma$, checking whether there exists a PPT operation that achieves the transition $\Lambda(\rho) = \sigma$ is an SDP problem.

\medskip

\textbf{\emph{Surpassing the limitations: discontinuous monotones, infinite sets, and resource catalysis.}} Does the result in Theorem~\ref{thm:Entanglement}
also hold if we take discontinuous monotones into account? As we will
see in the following, there exist resource theories which have a finite
complete set of resource monotones in this case, at least for qubit
systems. This holds for the theories of coherence and imaginarity
in the single-qubit setting. For the theory of coherence, all transformations
for a single qubit are described by the robustness of coherence $C_{R}$
and the $\Delta$-robustness of coherence $C_{\Delta,R}$, which are
given as~\citep{NapoliPhysRevLett.116.150502,PianiPhysRevA.93.042107,ChitambarPhysRevLett.117.030401,ChitambarPhysRevA.94.052336,ChitambarPhysRevA.95.019902,StreltsovPhysRevLett.119.140402}
\begin{align}
C_{R}\left(\rho\right) & =\min_{\tau}\left\{ s\geq0:\frac{\rho+s\tau}{1+s}\in\mathcal{I}\right\} ,\\
C_{\Delta,R}\left(\rho\right) & =\min_{\Delta[\sigma]=\Delta[\rho]}\left\{ s\geq0:\frac{\rho+s\sigma}{1+s}\in\mathcal{I}\right\} ,
\end{align}
where $\mathcal I$ is the set of incoherent states, i.e., states which are diagonal in a reference basis.
Note that in the single-qubit setting both measures can be evaluated
as $C_{R}(\rho)=2|\rho_{0,1}|$ and $C_{\Delta,R}(\rho)=|\rho_{0,1}|/\sqrt{\rho_{0,0}\rho_{1,1}}$~\citep{NapoliPhysRevLett.116.150502,ChitambarPhysRevA.94.052336,ChitambarPhysRevA.95.019902}.
From this we see that $C_{\Delta,R}(\rho)=1$ for all pure states
which have coherence. Since $C_{\Delta,R}(\rho)=0$ for all incoherent
states, this implies that $C_{\Delta,R}$ is not continuous.

For the resource theory of imaginarity we can construct a complete
set of monotones for single-qubit setting in terms of the Bloch coordinates
$(r_{x},r_{y},r_{z})$ of the states~\citep{WuPhysRevLett.126.090401,WuPhysRevA.103.032401}:
\begin{align}
I_{1}\left(\rho\right) & =r_{y}^{2},\\
I_{2}\left(\rho\right) & =\frac{r_{y}^{2}}{1-r_{x}^{2}-r_{z}^{2}}.
\end{align}
As has been shown in~~\citep{WuPhysRevLett.126.090401,WuPhysRevA.103.032401},
$I_{1}$ and $I_{2}$ do not increase under real operations, and fully
describe the transformations in the single-qubit setting. Moreover,
$I_{2}$ is not continuous, since $I_{2}(\rho)=1$ for all pure states
which have imaginarity and $I_{2}(\rho)=0$ on all real states. 

Note that in the resource theory of asymmetry~\citep{Gour_2008,Marvian_2014,LostaglioPhysRevX.5.021001} a complete set of monotones can also be constructed for single-qubit settings (see Supplemental Material for more details).  While the resource theory of quantum thermodynamics~\citep{Lostaglio_2019} in general does not contain resource-free pure states, we demonstrate in Supplemental Material that a complete set of monotones can also be found in quantum thermodynamics in the qubit setting.

Another way to surpass the limitations of Theorem~\ref{thm:Entanglement}
is to allow for an infinite set of resource monotones~\cite{TakagiPhysRevX.9.031053}.
The following is a simple construction of an infinite complete set of resource monotones:
\begin{equation}
R_{\nu}\left(\rho\right)=\inf_{\Lambda_{f}}\left\Vert \Lambda_{f}\left[\nu\right]-\rho\right\Vert _{1},\label{eq:CompleteSet}
\end{equation}
where $\nu$ is a quantum state which at the same time serves as a
parameter of the monotone $R_{\nu}$. To prove that $R_{\nu}$ is a resource monotone,
let $\tilde{\Lambda}_{f}$ be a free operation such that $R_{\nu}(\rho)\geq||\tilde{\Lambda}_{f}[\nu]-\rho||_{1}-\varepsilon$
for some $\varepsilon>0$ (note that such $\tilde{\Lambda}_{f}$ exists
for any $\varepsilon>0$). Then, for any free operation $\Lambda_{f}$
we find 
\begin{align}
R_{\nu}\left(\rho\right) & \geq\left\Vert \tilde{\Lambda}_{f}\left[\nu\right]-\rho\right\Vert _{1}-\varepsilon\geq\left\Vert \Lambda_{f}\circ\tilde{\Lambda}_{f}\left[\nu\right]-\Lambda_{f}\left[\rho\right]\right\Vert _{1}-\varepsilon\nonumber \\
 & \geq R_{\nu}\left(\Lambda_{f}\left[\rho\right]\right)-\varepsilon,
\end{align}
where we have used the fact that the trace norm does not increase under quantum operations. Since the above inequality holds true for any $\varepsilon>0$, we
conclude that $R_{\nu}(\rho)\geq R_{\nu}(\Lambda_{f}[\rho])$, as
claimed. To prove that $R_{\nu}$ form a complete set, consider two
states $\rho$ and $\sigma$ such that $R_{\nu}(\rho)\geq R_{\nu}(\sigma)$
for all states $\nu$. By choosing $\nu=\rho$ and noting that $R_{\rho}(\rho)=0$
it follows that $R_{\rho}(\sigma)=0$. This implies that $\rho$ can
be converted into $\sigma$ via free operations. The above arguments also imply that the set of all resource monotones is complete in any quantum resource theory, i.e., $\rho$ can be converted into $\sigma$ via free operations if and only if $R(\rho) \geq R(\sigma)$ for all resource monotones. We also note that different construction of a complete set of monotones for general quantum resource theories has been given in~\cite{TakagiPhysRevX.9.031053}.

A third way to surpass the limitations of Theorem~\ref{thm:Entanglement}
is to use quantum catalysis \citep{review_catalyst}. A quantum catalyst
is an additional quantum system which is not changed in the overall
procedure~\cite{JonathanPhysRevLett.83.3566}. Recently, significant progress has been achieved in the
study of correlated and approximate catalysis, where a catalyst can
build up correlations with the system, and the procedure is allowed
to have an error which can be made negligibly small \citep{SagawaPhysRevLett.126.150502,Varun_21,Bartosik_PhysRevLett.127.080502,WilmingPhysRevLett.127.260402,RubboliPhysRevLett.129.120506,datta_22}.
In this framework, a system state $\rho^{S}$ can be converted into
$\sigma^{S}$ if for any $\varepsilon>0$ there exists a catalyst
state $\tau^{C}$ and a free operation $\Lambda_{f}$ acting on the
system $S$ and the catalyst $C$ such that \citep{Varun_21,datta_22,RubboliPhysRevLett.129.120506,review_catalyst}
\begin{align}
\left\Vert \Lambda_{f}\left(\rho^{S}\otimes\tau^{C}\right)-\sigma^{S}\otimes\tau^{C}\right\Vert _{1} & \leq\varepsilon,\\
\mathrm{Tr}_{S}\left[\Lambda_{f}\left(\rho^{S}\otimes\tau^{C}\right)\right] & =\tau^{C}.
\end{align}
Remarkably, in the resource theory of coherence catalytic transformations
are completely described by a single quantity, known as the relative
entropy of coherence $C(\rho)=S(\Delta[\rho])-S(\rho)$ with the von
Neumann entropy $S(\rho)=-\mathrm{Tr}[\rho\log_{2}\rho]$. In particular,
it is possible to transform $\rho$ into $\sigma$ via dephasing covariant
operations and approximate catalysis if and only if $C(\rho)\geq C(\sigma)$~\cite{PhysRevLett.128.240501}, see also Supplemental Material. A similar
statement can be made for the resource theory of quantum thermodynamics
based on Gibbs-preserving operations. In this case, catalytic transformations
via Gibbs-preserving operations are fully described by the Helmholtz
free energy~\citep{SagawaPhysRevLett.126.150502}. Equivalently,
a catalytic transformation $\rho\rightarrow\sigma$ is possible in
this setting if and only if~\citep{SagawaPhysRevLett.126.150502}
$S(\rho||\gamma)\geq S(\sigma||\gamma)$ with the Gibbs state $\gamma$
and the quantum relative entropy $S(\rho||\gamma)=\mathrm{Tr}[\rho\log_{2}\rho]-\mathrm{Tr}[\rho\log_{2}\gamma]$.

\medskip

\textbf{\emph{Single complete resource monotone and total order.}}
One of the early problems in entanglement theory is to find a complete set of conditions that characterizes the state transformations.
While it quickly became clear that no finite set of conditions suffices~\citep{GourPhysRevA.72.022323} (except in special cases, e.g. pure states), the argument relies on specific features of entanglement.
Therefore the problem remains open for other resource theories.
In the last part of the article we will investigate the structure of ``simple'' resource theories which
have a single complete resource monotone, i.e., a free transformation
from $\rho$ to $\sigma$ is possible if and only if $R(\rho)\geq R(\sigma)$
for a single monotone $R$.
We will show that such theories are equivalent to total ordering of the states.
We will also provide some partial characterization of such theories, that shows such ``simple'' theories must have a very restricted form, which explains why state transformations in commonly considered resource theories cannot be governed by a single monotone.

In the following, we call a resource theory
\emph{totally ordered} if for any pair of states $\rho$ and $\sigma$
there exists a free transformation in (at least) one direction $\rho\rightarrow\sigma$
or $\sigma\rightarrow\rho$. We further introduce the resource monotone 
\begin{equation}
R(\rho)=\min_{\mu\in\mathcal{F}}\norm{\rho-\mu}_{1}, \label{eq:CompleteMonotone}
\end{equation}
where $\mathcal F$ is the set of free states. It is straightforward to see that $R$ is a monotone in any quantum resource theory.
We are now ready to prove the following theorem. 
\begin{thm} \label{thm:TotalOrderMonotone}
A resource theory has a single complete monotone if and only if the
theory is totally ordered. 
\end{thm}

The proof idea is to use the fact that $R\pqty{\pqty{1-\varepsilon} \sigma + \varepsilon \mu}$ is a strictly non-increasing function of $\varepsilon$ for any free state $\mu$.
We refer to Supplemental Material for the full proof. This shows that the existence of a single monotone that is complete
is equivalent to a total ordering of the set of states by the free
transformations. We will now prove some additional features of totally ordered resource theories.

\begin{thm}\label{totally_ordered}
Any totally ordered quantum resource theory allows for free transformations
between any two pure states $\ket{\psi}\rightarrow\ket{\phi}$. 
\end{thm}

See Supplemental Material for the full proof. It is important to note that Theorem~\ref{totally_ordered} implies 
\begin{equation}
R(\psi)=R(\phi)\label{eq:TotalOrderPure}
\end{equation}
for any two pure states $\ket{\psi}$ and $\ket{\phi}$.

We will now fully characterize all totally ordered resource theories for $d=2$.
We will start by characterizing the set of free states, using
again the monotone $R$ in Eq.~(\ref{eq:CompleteMonotone}). Note for two single-qubit states
$\rho$ and $\sigma$ with Bloch vectors $\boldsymbol{r}$ and $\boldsymbol{s}$
it holds $||\rho-\sigma||_{1}=|\boldsymbol{r}-\boldsymbol{s}|$. Since
all pure states are equally far away from the set of free states due
to Eq.~(\ref{eq:TotalOrderPure}), it must be that the set of free states is a ball around
the maximally mixed state. Denoting the radius of this ball by $t$
we can characterize the set of free states as follows:
\begin{equation}
\mathcal{F}_t=\left\{ \sigma:\left\Vert \sigma-\frac{\openone}{2}\right\Vert _{1}\leq t\right\} .\label{eq:FtotallyOrdered}
\end{equation}
with $t \in [0,1]$. For any given $t$ we can now evaluate the resource monotone $R$
for any state $\rho$:
\begin{equation}
R(\rho)=\max\{|\boldsymbol r|-t,0\}.
\end{equation}
Thus, in a totally ordered resource
theory for a single qubit all state transformations are determined
by the length of the Bloch vector. For any two resource states $\rho$
and $\sigma$ (with Bloch vectors $\boldsymbol{r}$ and $\boldsymbol{s}$)
a free transformation $\rho\rightarrow\sigma$ is possible if and
only if $|\boldsymbol{r}|\geq|\boldsymbol{s}|$. Moreover, a transformation
$\rho\rightarrow\sigma$ is always possible whenever $|\boldsymbol{s}|\leq t$,
since $\sigma$ is a free state in this case. 

An example for a totally ordered resource theory in the single-qubit
setting is the resource theory of purity~\cite{HorodeckiPhysRevA.67.062104,GOUR20151}, which corresponds to the
case $t=0$. We will now show that a totally ordered resource theory exists for any $t \in [0,1]$. For a given $t$, we define the set of free operations to be all unital operations, i.e., all operations with the property $\Lambda[\openone/2] = \openone/2$. Additionally, all fixed-output operations such that $\Lambda[\rho] = \sigma$ with $\sigma \in \mathcal F_t$ are considered free. Noting that via unital operations it is possible to transform a qubit state $\rho$ into another qubit state $\sigma$ if and only if $|\boldsymbol r| \geq |\boldsymbol s|$~\cite{GOUR20151}, we see that the free states and operations defined in this way give rise to a totally ordered resource theory, with $\mathcal F_t$ being the set of free states.
Note that the key property enabling this construction is that unital operations induce a total order.
However since this property does not hold for $d \geq 3$, we cannot generalize the construction to higher dimensional systems.

\medskip
\textbf{\emph{Conclusions.}} We have investigated the possibility to have a complete set of monotones in general quantum resource theories. Using only minimal assumptions, such as monotonicity and continuity, we have proven that a complete finite set of monotones does not exist, if a resource theory contains free pure states. This result is applicable to the theory of entanglement in bipartite and multipartite settings, and also to the theories of coherence, imaginarity, and asymmetry. It is however possible to find complete sets of monotones by either allowing discontinuity, or considering infinite sets, and we gave examples for such complete sets in various resource theories.

We have further considered resource theories where the state transformations are governed by a single monotone.
We proved that any such theory must be totally ordered, where any pair of states admit a free transformation in (at least) one direction.
We provided a partial characterization of any such theory: any totally ordered resource theory must allow for free transformations between all pure states, and provided a full characterization of state transformations for all totally ordered resource theories for a single qubit.
It remains an open question whether there exist totally ordered resource theories for $d \geq 3$.
Nevertheless, this shows the severe restrictions one imposes when we assume that transformations are governed by a single monotone. Another open problem concerns the extension of our results to the resource theories of quantum channels, where -- instead of states -- transformations between quantum channels are considered~\cite{PhysRevLett.122.190405}. It is not clear at this moment how the results presented in this article extend to these resource theories.

\textbf{\emph{Acknowledgements.}} We acknowledge discussion with Ludovico Lami, Bartosz Regula, Henrik Wilming, and Andreas Winter. This work was supported by the ``Quantum Optical Technologies'' project, carried out within the International Research Agendas programme of the Foundation for Polish Science co-financed by the European Union under the European Regional Development Fund, the ``Quantum Coherence and Entanglement for Quantum Technology'' project, carried out within the First Team programme of the Foundation for Polish Science co-financed by the European Union under the European Regional Development Fund, and the National Science Centre, Poland, within the QuantERA II Programme (No 2021/03/Y/ST2/00178, acronym ExTRaQT) that has received funding from the European Union's Horizon 2020 research and innovation programme under Grant Agreement No 101017733.

\bibliography{literature}

\clearpage
\section*{Supplemental Material}

\subsection*{Complete set of monotones for single-qubit state transformations in the resource theory of asymmetry}

In the resource theory of asymmetry \citep{Gour_2008,Marvian_2014,LostaglioPhysRevX.5.021001},
we can also construct a complete set of monotones for single-qubit
state transformations. For a given initial qubit state $\rho$, the
achievable set of qubit states $\{\sigma\}$ can be found using the
following condition \citep{LostaglioPhysRevX.5.021001}: 
\begin{equation}
|\sigma_{0,1}|\leq|\rho_{0,1}|\sqrt{\chi},\label{ti_lostaglio}
\end{equation}
where $\chi=\min\{\sigma_{0,0}/\rho_{0,0},(1-\sigma_{0,0})/(1-\rho_{0,0})\}$.
Considering Eq. (\ref{ti_lostaglio}), we can easily construct
the following complete set of monotones: 
\begin{eqnarray}
A_{1}(\rho) & = & \frac{|\rho_{0,1}|}{\sqrt{\rho_{0,0}}},\label{TI1}\\
A_{2}(\rho) & = & \frac{|\rho_{0,1}|}{\sqrt{1-\rho_{0,0}}}.\label{TI2}
\end{eqnarray}
It is easy to check that the monotones are not continuous. To see
this, consider a pure state with $\rho_{0,1}=\varepsilon\sqrt{1-\varepsilon^{2}}$
and $\rho_{0,0}=\varepsilon^{2}$, where $\varepsilon>0$. Note that $A_{1}(\rho)$ converges to one in the limit $\varepsilon \rightarrow 0$, whereas it is
zero for the non-resourceful states. In a similar way, by taking into account a pure state with $\rho_{0,1}=\varepsilon\sqrt{1-\varepsilon^{2}}$
and $\rho_{1,1}=\varepsilon^{2}$, we can verify that $A_2(\rho)$ is discontinuous.

\subsection*{Complete set of monotones for thermal operations for a single qubit}
Notably, in qubit scenarios, the resource theory of thermodynamics can be seen as a combination of the resource theory of asymmetry and the resource theory of Gibbs-preserving operations, i.e., a resource theory where all operations which preserve the Gibbs state are free \citep{CwiklinskiPhysRevLett.115.210403}. The complete criteria for state transformations under Gibbs-preserving
operations are known \citep{Buscemi_PhysRevA.95.012110}, i.e., a
state $\rho$ can be transformed into $\sigma$ under Gibbs-preserving
operation if and only if the following inequalities hold:
\begin{eqnarray}
 &  & D_{\max}(\rho||\gamma)\geq D_{\max}(\sigma||\gamma),\\
 &  & D_{\max}(\gamma||\rho)\geq D_{\max}(\gamma||\sigma),\\
 &  & D_{\min}(\gamma||\rho)\geq D_{\min}(\gamma||\sigma).
\end{eqnarray}
Here, $\gamma$ represents Gibbs state, $D_{\max}(\rho_{1}||\rho_{2})=\log\min\{\lambda:\rho_{1}\leq\lambda\rho_{2}\}$
and $D_{\min}(\rho_{1}||\rho_{2})=-\log\Tr(\Pi_{\rho_{1}}\rho_{2})$
are the max- and min-relative entropies respectively \citep{NDatta_2009}, and $\Pi_{\rho_1}$ denotes the projector onto the support of $\rho_1$. These three monotones, when combined with the monotones stated in Eqs. (9) and (10) of the main text, provide a complete set of monotones for thermal operations in the qubit scenario \citep{CwiklinskiPhysRevLett.115.210403}. 

\subsection*{Approximate catalysis and the resource theory of coherence}
Here, we will prove that in the resource theory of coherence based
on dephasing covariant incoherent operations (DIO) an approximate catalytic transformation
from $\rho$ into $\sigma$ is possible if and only if 
\begin{equation}
C(\rho)\geq C(\sigma),\label{eq:CoherenceCatalysis}
\end{equation}
where $C(\rho)=S(\Delta[\rho])-S(\rho)$ is the relative entropy of
coherence.

The proof follows similar lines of reasoning as the proofs for the
resource theories of entanglement and thermodynamics~\cite{SagawaPhysRevLett.126.150502,Varun_21,datta_22}, and a discussion for general quantum resource theories has also been presented in~\cite{PhysRevLett.128.240501}. In
particular, we first show that a transformation with approximate catalysis
is possible whenever $\rho$ can be converted into $\sigma$ via DIO in the
asymptotic setting with rate at least one. An asymptotic
transformation with rate at least one is possible if and only if for
any $\varepsilon>0$ and any $\delta>0$ there exist integers $n$
and $m$ with $n>m$ and a dephasing covariant operation $\Lambda$
such that \begin{subequations}\label{eq:asymptotic}
\begin{align}
\left\Vert \Lambda\left[\rho^{\otimes n}\right]-\sigma^{\otimes m}\otimes\ket{0}\!\bra{0}^{\otimes n-m}\right\Vert _{1} & \leq\varepsilon,\\
\frac{m}{n}+\delta & \geq1.
\end{align}
\end{subequations} We define the state $\Gamma=\Lambda\left[\rho^{\otimes n}\right]$,
which is a quantum state on $n$ copies of the system $S_{1}\otimes S_{2}\otimes\cdots\otimes S_{n}$.
We further define $\Gamma_{i}$ to be the reduced state of $\Gamma$
on $S_{1}\otimes S_{2}\otimes\cdots\otimes S_{i}$, and we set
$\Gamma_{0}=1$. Consider now a catalyst state of the form 
\begin{equation}
\tau=\frac{1}{n}\sum_{k=1}^{n}\rho^{\otimes k-1}\otimes\Gamma_{n-k}\otimes\ket{k}\!\bra{k}.\label{eq:tau}
\end{equation}
The state of the catalyst acts on a Hilbert space of $S^{\otimes{n-1}}\otimes K$,
where $n$ is chosen such that Eqs.~(\ref{eq:asymptotic}) are fulfilled.
Moreover, $K$ is a system of dimension $n$, with incoherent states
$\{\ket{k}\}$. The initial system $S$ will be denoted as $S_{1}$,
and $n-1$ copies of $S$ which are part of the catalyst are denoted
by $S_{2},\ldots,S_{n}$. Thus, the system $C$ of the catalyst is
composed of $C=S_{2}\ldots S_{n}K$. 

Consider now the following 3-step procedure acting on the system and the
catalyst:
\begin{enumerate}
\item A von Neumann measurement is applied on the register $K$ in the basis
$\{\ket{k}\}$. If the measurement outcome is $n$, the DIO $\Lambda$
defined in Eqs.~(\ref{eq:asymptotic}) is applied onto the systems
$S_{1}\otimes S_{2}\otimes\cdots\otimes S_{n}$. In this way, a state
$\nu=\nu^{S_{1}\otimes S_{2}\otimes\cdots\otimes S_{n}K}$ on $S_{1}\otimes S_{2}\otimes\cdots\otimes S_{n}K$
is transformed as follows:
\begin{align}
\nu & \rightarrow\sum_{k=1}^{n-1}p_{k}\nu_{k}^{S_{1}\otimes S_{2}\otimes\cdots\otimes S_{n}}\otimes\ket{k}\!\bra{k}^{K}\\
 & +p_{n}\Lambda\left[\nu_{n}^{S_{1}\otimes S_{2}\otimes\cdots\otimes S_{n}}\right]\otimes\ket{n}\!\bra{n}^{K},\nonumber 
\end{align}
where we defined the probability
\begin{align}
p_{k} =\mathrm{Tr}\left[\nu\openone^{S_{1}\otimes S_{2}\otimes\cdots\otimes S_{n}}\otimes\ket{k}\!\bra{k}^{K}\right],
\end{align}
and for $p_k > 0$ the states $\nu_{k}^{S_{1}\otimes S_{2}\otimes\cdots\otimes S_{n}}$ are defined as 
\begin{equation}
    \nu_{k}^{S_{1}\otimes S_{2}\otimes\cdots\otimes S_{n}} =\frac{1}{p_{k}}\mathrm{Tr}_{K}\left[\nu\openone^{S_{1}\otimes S_{2}\otimes\cdots\otimes S_{n}}\otimes\ket{k}\!\bra{k}^{K}\right].
\end{equation}
Recalling that $\Lambda$ is dephasing covariant, it is straightforward
to verify that the overall transformation on $S_{1}\otimes S_{2}\otimes\cdots\otimes S_{n}K$
is also dephasing covariant.
\item A dephasing covariant unitary is applied on the register $K$ such
that $\ket{n}\rightarrow\ket{1}$ and $\ket{i}\rightarrow\ket{i+1}$. 
\item A SWAP unitary is applied on the subsystems, which shifts $S_{i}\rightarrow S_{i+1}$
and $S_{n}\rightarrow S_{1}$.
\end{enumerate}
Note that the overall transformation described in the steps 1-3
is dephasing covariant. 

We will now analyze how this transformation acts on a total system-catalyst
state, which initially has the form
\begin{equation}
\rho\otimes\tau=\frac{1}{n}\sum_{k=1}^{n}\rho^{\otimes k}\otimes\Gamma_{n-k}\otimes\ket{k}\!\bra{k}.
\end{equation}
Applying 1. step of the protocol we obtain
\begin{equation}
\mu_{1}=\frac{1}{n}\sum_{k=1}^{n-1}\rho^{\otimes k}\otimes\Gamma_{n-k}\otimes\ket{k}\!\bra{k}+\frac{1}{n}\Gamma\otimes\ket{n}\!\bra{n}.
\end{equation}
In the 2. step of the protocol, $\mu_{1}$ is transformed into 
\begin{equation}
\mu_{2}=\frac{1}{n}\sum_{k=1}^{n}\rho^{\otimes k-1}\otimes\Gamma_{n+1-k}\otimes\ket{k}\!\bra{k}.\label{eq:Mu2}
\end{equation}
It is straightforward to check that tracing out $S_{n}$ from $\mu_{2}$
gives $\tau$, which is the initial state of the catalyst, see Eq.~(\ref{eq:tau}).
In 3. step of the protocol, the state $\mu_{2}$ is transformed into
the final state $\mu^{SC}$ such that $\mathrm{Tr}_{S}[\mu^{SC}]=\tau^C$. 

We will now show that for any $\varepsilon>0$ and any $\delta>0$
the protocol can be performed such that
\begin{equation}
\left\Vert \mu^{SC}-\sigma^{S}\otimes\tau^{C}\right\Vert _{1}<2\left(\varepsilon+\delta\right).\label{eq:DecouplingMain}
\end{equation}
Note that $\mu^{SC}$ is equivalent to the state $\mu_{2}$ in Eq.~(\ref{eq:Mu2})
up to a cyclic SWAP. This implies that 
\begin{equation}
\left\Vert \mu^{SC}-\sigma^{S}\otimes\tau^{C}\right\Vert _{1}=\left\Vert \mu_{2}-\gamma\right\Vert _{1},
\end{equation}
 where the state $\gamma$ is defined as 
\begin{equation}
\gamma=\frac{1}{n}\sum_{k=1}^{n}\rho^{\otimes k-1}\otimes\Gamma_{n-k}\otimes\sigma\otimes\ket{k}\!\bra{k}.
\end{equation}
Now, we obtain 
\begin{align}
\Vert\mu_{2}-\gamma\Vert_{1} & =\frac{1}{n}\sum_{k=1}^{n}\Vert\Gamma_{n+1-k}-\Gamma_{n-k}\otimes\sigma\Vert_{1}\\
 & =\frac{1}{n}\sum_{k=1}^{n-m}\Vert\Gamma_{n+1-k}-\Gamma_{n-k}\otimes\sigma\Vert_{1}\nonumber \\
 & +\frac{1}{n}\sum_{k=n-m+1}^{n}\Vert\Gamma_{n+1-k}-\Gamma_{n-k}\otimes\sigma\Vert_{1}\nonumber \\
 & \leq2\delta+\frac{1}{n}\sum_{l=1}^{m}\Vert\Gamma_{l}-\Gamma_{l-1}\otimes\sigma\Vert_{1}\nonumber \\
 & \leq2\delta+\frac{1}{n}\sum_{l=1}^{m}\Vert\Gamma_{l}-\sigma^{\otimes l}\Vert_{1}\nonumber \\
 & +\frac{1}{n}\sum_{l=1}^{m}\Vert\Gamma_{l-1}-\sigma^{\otimes l-1}\Vert_{1}\nonumber \\
 & \leq2(\delta+\frac{m}{n}\varepsilon)\leq2(\delta+\varepsilon).\nonumber 
\end{align}
In the first inequality we used Eqs.~(\ref{eq:asymptotic}) together
with the inequality $||\rho-\sigma||_{1}\leq2$ for any $\rho$ and $\sigma$. In the second inequality
we used the triangle inequality. The third inequality follows again
from Eqs. (\ref{eq:asymptotic}). The above arguments prove that $\rho$ can be converted into $\sigma$ via DIO with approximate catalysis whenever an asymptotic conversion via DIO is possible with rate at least one.

As follows from results in~\cite{ChitambarPhysRevA.97.050301}, it is possible to convert $\rho$
into $\sigma$ via asymptotic DIO with rate at least one whenever
Eq.~(\ref{eq:CoherenceCatalysis}) is fulfilled. This implies that Eq.~(\ref{eq:CoherenceCatalysis}) also guarantees that 
the transformation $\rho\rightarrow\sigma$ is possible via DIO with
approximate catalysis.

To show that a transformation is not possible when Eq.~(\ref{eq:CoherenceCatalysis})
is violated, we will now prove that the relative entropy of coherence cannot increase under DIO with approximate catalysis. In particular, if $\rho$ can be converted into
$\nu$ via DIO with approximate catalysis, it must hold that 
\begin{equation}
C(\rho)\geq C(\nu).
\end{equation}
The proof follows very similar lines of reasoning as the proof for
bipartite pure states in entanglement theory~\cite{Varun_21}, we present
it below for completeness.

Note that the relative entropy of coherence in bipartite systems fulfills~\cite{Xi2015}
\begin{equation}
C(\rho^{AB})\geq C(\rho^{A})+C(\rho^{B})
\end{equation}
with equality when $\rho^{AB}=\rho^{A}\otimes\rho^{B}$. Assume now
that for any $\varepsilon>0$ there exists a catalyst state $\tau^{C}$
and a DIO $\Lambda$ acting on $SC$ such that the final state $\sigma^{SC}=\Lambda(\rho^{S}\otimes\tau^{C})$
has the properties 
\begin{align}
\left\Vert \mathrm{Tr}_{C}\left[\sigma^{SC}\right]-\nu^{S}\right\Vert _{1} & <\varepsilon,\\
\mathrm{Tr}_{S}\left[\sigma^{SC}\right] & =\tau^{C}.
\end{align}
Using the properties of the relative entropy of coherence we obtain
\begin{equation}
C\left(\sigma^{SC}\right)\leq C\left(\rho^{S}\right)+C\left(\tau^{C}\right)\label{1ineq}
\end{equation}
and also 
\begin{equation}
C\left(\sigma^{SC}\right)\geq C\left(\mathrm{Tr}_{C}\left[\sigma^{SC}\right]\right)+C\left(\tau^{C}\right).\label{2ineq}
\end{equation}
From Eqs. (\ref{1ineq}) and (\ref{2ineq}) it follows 
\begin{equation}
C\left(\rho^{S}\right)\geq C\left(\mathrm{Tr}_{C}\left[\sigma^{SC}\right]\right).
\end{equation}
If $||\mathrm{Tr}_{C}[\sigma^{SC}]-\nu^{S}||_{1}$ can
be made arbitrarily small, then by continuity of the relative entropy
of coherence~\cite{PhysRevLett.116.120404} we get $C(\rho^{S})\geq C(\nu^{S})$, and
the proof is complete. 

\subsection*{Full rank states cannot be converted into pure resource states}\label{full_rank_state}
Here we will prove that it is not possible to convert a full rank
state $\rho$ into a pure resource state $\ket{\phi}$ via free operations, a proof of this can also be found in~\cite{FangPhysRevLett.125.060405,PhysRevA.101.062315}.
By contradiction, suppose that such a conversion is possible, i.e.,
for any $\varepsilon>0$ there is a free operation $\Lambda_{f}$
such that 
\begin{equation}
\braket{\phi|\Lambda_{f}\left[\rho\right]|\phi}>1-\varepsilon.\label{eq:FullRank}
\end{equation}
Recalling that $\rho$ has full rank, for any
$\ket{\psi}\in\mathcal{H}_{d}$ there exists some state $\sigma$
such that 
\begin{equation}
\rho=p_{\min}\psi+(1-p_{\min})\sigma,
\end{equation}
where $p_{\min}$ is the smallest eigenvalue of $\rho$, see also Lemma~\ref{lem:LemmaDecomposition} below.
With this, we obtain 
\begin{align}
\braket{\phi|\Lambda_{f}\left[\rho\right]|\phi} & =p_{\min}\braket{\phi|\Lambda_{f}\left[\psi\right]|\phi}+(1-p_{\min})\braket{\phi|\Lambda_f\left[\sigma\right]|\phi}\nonumber \\
 & \leq1-p_{\min}\left(1-\braket{\phi|\Lambda_{f}\left[\psi\right]|\phi}\right).
\end{align}
Together with Eq.~(\ref{eq:FullRank}) we obtain
\begin{equation}
\braket{\phi|\Lambda_{f}\left[\psi\right]|\phi}>1-\frac{\varepsilon}{p_{\min}}.
\end{equation}
Since we can choose $\varepsilon>0$ arbitrarily small, we conclude
that for any pure state $\ket{\psi}\in\mathcal{H}_{d}$ there is a
free operation transforming $\ket{\psi}$ into $\ket{\phi}$. By linearity,
this extends also to any mixed state on $\mathcal{H}_{d}$, i.e.,
any mixed state can be transformed into $\ket{\phi}$ via free operations.
Noting that this also applies to any free state $\sigma_{f}$ we arrive
at a contradiction and the proof is complete.

We will now provide a proof of a statement which has been used above in the proof.
\begin{lem} \label{lem:LemmaDecomposition}
Let $p_{\min}$ be the smallest eigenvalue of $\rho$. For any $\ket{\psi}\in\mathcal{H}_{d}$
there exists some state $\sigma$ such that
\begin{equation}
\rho=p_{\min}\psi+(1-p_{\min})\sigma.\label{eq:LemmaDecomposition}
\end{equation}
\end{lem}
\begin{proof}
Noting that $\braket{\phi|\rho|\phi}\geq p_{\min}$ is true for any
$\ket{\phi}\in\mathcal{H}_{d}$, it must be that 
\begin{equation}
\rho-p_{\min}\psi\geq0
\end{equation}
for any $\ket{\psi}\in\mathcal{H}_{d}$. Since $p_{\min}\leq1/d<1$,
we can define the state
\begin{equation}
\sigma=\frac{\rho-p_{\min}\psi}{1-p_{\min}},
\end{equation}
such that Eq.~(\ref{eq:LemmaDecomposition}) is fulfilled. 
\end{proof}

\subsection*{Proof of Theorem 2}

If there is a single monotone that is complete, then for any two states
$\rho,\sigma$ either we have $\rho\rightarrow\sigma$ possible by
free operations or $\sigma\rightarrow\rho$. This means that the ordering
on the set of states induced by free transformation is a total order.

To prove the converse, assume that the free transformations induce
a total order on the set of states. We will show that the monotone
defined in Eq.~(\ref{eq:CompleteMonotone}) of the main text is complete. Since
$R$ is monotonically nonincreasing under free operations, it follows
that a free transformation from $\sigma$ into $\rho$ is impossible
whenever $R(\rho)>R(\sigma)$. Since the resource theory is totally
ordered, it must be that a free transformation $\rho\rightarrow\sigma$
is always possible in this case. It remains to consider the case $R(\rho)=R(\sigma)$.
If $R(\rho)=0$, both $\rho$ and $\sigma$ are free states which
can be interconverted via free operations.
If $R(\rho)>0$, we will
show that we can transform $\rho$ arbitrarily close to $\sigma$,
i.e., for any $\delta>0$ there is a free operation $\Lambda_{f}$
such that $||\Lambda_{f}[\rho]-\sigma||_{1}<\delta$. Let us define
$\sigma_{\varepsilon}=\pqty{1-\varepsilon}\sigma+\varepsilon\mu_{f}$,
where $\mu_{f}\in\mathcal{F}$ achieves the minimum distance from
$\sigma$ to the set of free states. We obtain 
\begin{equation}
R(\sigma_{\varepsilon})\leq\left\Vert \sigma_{\varepsilon}-\mu_{f}\right\Vert _{1}=(1-\varepsilon)\left\Vert \sigma-\mu_{f}\right\Vert _{1}=(1-\varepsilon)R(\sigma),
\end{equation}
and thus $R(\sigma_{\varepsilon})<R(\rho)$ for all $\varepsilon>0$.
Again recalling that the resource theory is totally ordered, there
is a free transformation from $\rho$ into $\sigma_{\varepsilon}$
for all $\varepsilon>0$. The proof is complete by noting that $||\sigma_{\varepsilon}-\sigma||_{1}$
can be made arbitrarily small by choosing small enough $\varepsilon$.

\subsection*{Proof of Theorem 3}

Consider a resource monotone of the form~Eq. (\ref{eq:CompleteMonotone}).
As explained in the proof of Theorem~2,
for a totally ordered resource theory this monotone determines all
state transformations. We will now prove that 
\begin{equation}
R(\psi)=R(\phi)\label{eq:TotalOrderPure1}
\end{equation}
must hold for all pure states. By contradiction, assume that there
exist two pure states such that $R(\psi)>R(\phi)>0$. Consider the full
rank state $\rho_{\varepsilon}=(1-\varepsilon)\psi+\varepsilon\openone/d$
with $0<\varepsilon<1$. By continuity, it must be that $R(\rho_{\varepsilon})>R(\phi)$
for small enough $\varepsilon$. Recalling that $R$ fully determines
all state transformations, there exists a free transformation from
$\rho_{\varepsilon}$ into $\phi$. This is not possible, since $\rho_{\varepsilon}$
is a full rank state, and $\ket{\phi}$ is a pure resource state~\cite{FangPhysRevLett.125.060405}, see the section containing Lemma~\ref{lem:LemmaDecomposition}. Using
again the fact that $R$ determines all state transformations, Eq.~(\ref{eq:TotalOrderPure1})
implies that there are free transformations between any pair of pure
states, as claimed.

\end{document}